\documentclass[preprintnumbers,aps,amsmath,amssymb,pra,twocolumn,showkeys,superscriptaddress,floatfix]{revtex4}
\usepackage[dvipdfmx]{graphicx}
\usepackage{amsmath, amssymb, amsthm}
\usepackage{mathtools}
\usepackage{color}
\usepackage{listings}

\def\Bra#1{\left\langle#1\right|}
\def\Ket#1{\left|#1\right\rangle}

\theoremstyle{definition}
\newtheorem{theorem}{Theorem}

\newtheorem{proposition}[theorem]{Proposition}
\theoremstyle{definition}

\theoremstyle{remark}

\begin{document}

\title{Multipartite entanglement outperforming bipartite entanglement\\ under limited quantum system sizes}

\author{Hayata Yamasaki}
\email{yamasaki@eve.phys.s.u-tokyo.ac.jp}
\affiliation{Department of Physics, Graduate School of Science, The University of Tokyo, 7--3--1 Hongo, Bunkyo-ku, Tokyo, Japan}

\author{Alexander Pirker}
\affiliation{Institut f\"ur Theoretische Physik, Universit\"at Innsbruck, Technikerstra{\ss}e 21a, 6020 Innsbruck, Austria}

\author{Mio Murao}
\affiliation{Department of Physics, Graduate School of Science, The University of Tokyo, 7--3--1 Hongo, Bunkyo-ku, Tokyo, Japan}

\author{Wolfgang D\"{u}r}
\affiliation{Institut f\"ur Theoretische Physik, Universit\"at Innsbruck, Technikerstra{\ss}e 21a, 6020 Innsbruck, Austria}

\author{Barbara Kraus}
\affiliation{Institut f\"ur Theoretische Physik, Universit\"at Innsbruck, Technikerstra{\ss}e 21a, 6020 Innsbruck, Austria}

\date{\today}

\begin{abstract}
    Multipartite quantum entanglement serves as a resource for spatially separated parties performing distributed quantum information processing. Any multipartite entangled state can be generated from appropriately distributed bipartite entangled states by local operations and classical communication (LOCC), and in this sense, any distributed process based on shared multipartite entanglement and LOCC is simulatable by using only bipartite entangled states and LOCC\@. We show here that this reduction scenario does not hold when there exists a limitation on the size of the local quantum system of each party. Under such a limitation, we prove that there exists a set of multipartite quantum states such that these states in the set \textit{cannot} be prepared from any distribution of bipartite entanglement while the states \textit{can} be prepared from a common resource state exhibiting multipartite entanglement. We also show that temporal uses of bipartite quantum communication resources within a limitation of local system sizes are sufficient for preparing this common resource state exhibiting multipartite entanglement, yet there also exist other states exhibiting multipartite entanglement which cannot be prepared even in this setting. Hence, when the local quantum system sizes are limited, multipartite entanglement is an indispensable resource without which certain processes still cannot be accomplished.
\end{abstract}

\pacs{}
\keywords{multipartite entanglement, limitation on quantum system size, local operations and classical communication (LOCC)}

\maketitle

\section{Introduction}

Multipartite quantum entanglement ubiquitously appears in many-body quantum systems in condensed matter physics~\cite{A2} and quantum gravity~\cite{R4}, and also serves as a resource for multiparty tasks in distributed quantum information processing such as measurement-based quantum computation~\cite{R1,R2,R3}, distributed sensing~\cite{K1,E4}, and quantum networking~\cite{P3}.
Such a distributed setting is also considered as a promising candidate for realizing large-scale quantum computation due to technological limitations on the number of low-noise qubits which can be stored in a single quantum device.
In a distributed setting where spatially separable parties can freely perform local operations and classical communication (LOCC), any multipartite entangled state can be prepared by LOCC from initially distributed bipartite entangled states among the parties, using quantum teleportation~\cite{B4}.
In this regard, even if multipartite entanglement is used for a task, initially sharing bipartite entangled states is sufficient,
and hence, it would be natural to doubt whether multipartite entanglement is necessary for performing tasks by LOCC\@.

\begin{figure}[!t]
    \centering
    \includegraphics[width=0.79\linewidth]{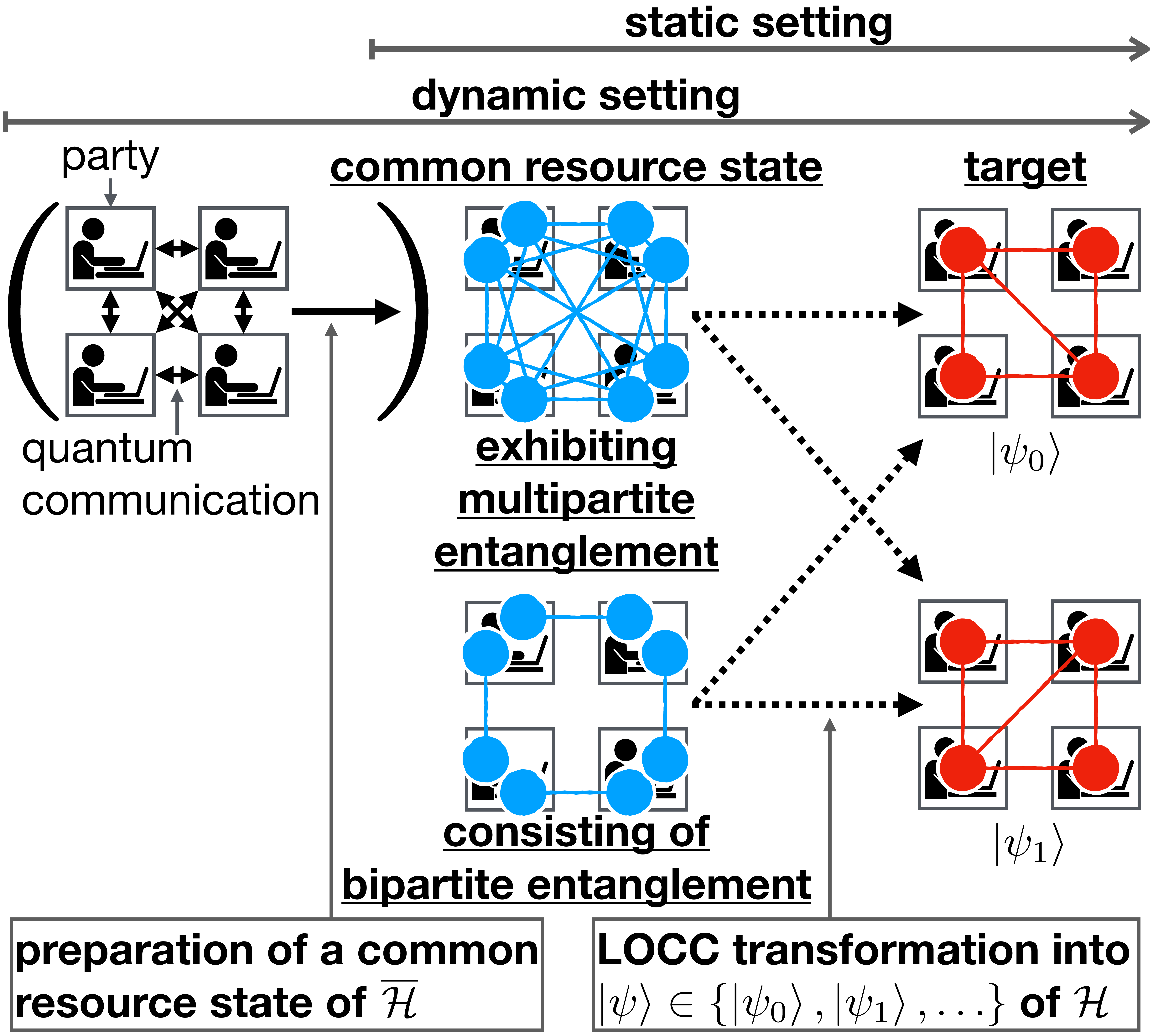}
\caption{The task of system-size-limited quantum state transformation, where the parties transform a common resource state represented by blue circles by LOCC into an arbitrary state in a given target set $\left\{\Ket{\psi_0},\Ket{\psi_1},\ldots\right\}$ represented by red circles. To differentiate the capabilities of common resource states exhibiting multipartite entanglement at the top and those consisting only of bipartite entanglement at the bottom, where each connected pair of blue circles represents a bipartite entangled state, we consider the static setting where each party's local system size for storing the common resource state is limited. We also consider the dynamic setting where the parties have to prepare a common resource state within these limitations by performing quantum communication, in addition to storing the common resource state. The difference in the capabilities arises in terms of achievability of this task.}
\label{fig:intro}
\end{figure}

In this paper, we show \textit{nontrivial} examples demonstrating the difference between entangled resource states consisting only of bipartite entanglement and those exhibiting multipartite entanglement.
This difference arises when there exists a limitation on the size of each party's local quantum system, that is, the dimension of the Hilbert space representing the local quantum system.
Our comparison between bipartite and multipartite entanglement, motivated by technological limitations on the number of qubits which can be stored in one quantum device, differs from the comparison in the context of quantum key distribution~\cite{E3,P4} as we consider the cost of LOCC to be negligible.
The difference can also be observed in a trivial example of qubits as follows.
Consider three parties $A$, $B$, and $C$ sharing two Bell states ${\left(\frac{1}{\sqrt{2}}\left(\Ket{00}+\Ket{11}\right)\right)}^{\otimes 2}$, one of which is between $A$ and $B$, and the other of which is between $B$ and $C$.
These two Bell states as a whole are regarded as a state consisting of bipartite entangled states.
In this case, once these two Bell states are given to the parties, the parties can transform the two Bell states by LOCC into any three-qubit state shared among $A$, $B$, and $C$, such as
the Greenberger-Horne-Zeilinger (GHZ) state $\Ket{\textup{GHZ}}\coloneqq\frac{1}{\sqrt{2}}\left(\Ket{000}+\Ket{111}\right)$ and the $W$ state $\Ket{W}\coloneqq\frac{1}{\sqrt{3}}\left(\Ket{100}+\Ket{010}+\Ket{001}\right)$, which we regard as states exhibiting multipartite entanglement.
However, if each party's local system size is limited to one qubit,
the parties cannot store any state consisting of a collection of bipartite entangled states to obtain $\Ket{\textup{GHZ}}$ and $\Ket{W}$ by LOCC while the parties can still store these states exhibiting multipartite entanglement as a resource for performing some task by LOCC\@.

Apart from the above trivial example of qubits,
this paper aims to demonstrate the difference even in cases where the size of local systems of some parties is not limited to one qubit.
Given an entangled state transformable into another entangled state, the former state can be considered to have more capability as a resource than the latter state.
If such a resource state having more capability is shared among parties, the parties may transform the shared resource state by LOCC into a suitable form for performing a given task.
This paradigm yields a \textit{common resource state}~\cite{S3,G3} transformable into any state in a given set, that is, a resource state having more capability than any state in the set, such as the two Bell states for the set $\left\{\Ket{\textup{GHZ}},\Ket{W}\right\}$.
We call this set of states the \textit{target set}.
Similarly, Ref.~\cite{H3} also introduces common resource states in terms of state convertibility by stochastic LOCC\@.

As illustrated in Fig.~\ref{fig:intro}, we consider two settings of state preparation tasks for differentiating capabilities of common resource states consisting only of a collection of bipartite entangled states and those exhibiting multipartite entanglement.
We call the tasks \textit{system-size-limited quantum state preparation},
where one of the two settings is called a \textit{static} setting, and the other is called a \textit{dynamic} setting.
In the static setting, we analyze each party's local system size for storing a common resource state for a given target set.
For a given target set of states of a multipartite system in general, there may not exist any common resource state in the multipartite system itself transformable by LOCC into all the states in the set.
In particular, given a multipartite system where each local dimension is $d$, almost no LOCC transformation among pure states of the system is possible~\cite{V2,S1,S2,H2,G2,S4}.
This fact implies that, in general, a common resource state for a set of multipartite states may be a state of a higher-dimensional system than that for the set itself.
If there is a limitation on each party's local system size, it may not be possible for the parties to store an entangled state of a higher-dimensional system serving as a common resource state.
Despite the efforts to understand properties of multipartite entanglement~\cite{H1,P1,E1,E2,W1,B2}, general quantitative conditions of the smallest system size for common resource states have not yet been established.
In this paper, we provide nontrivial instances where, within a given limitation on local system sizes, the preparation of a state in a given target set is \textit{not} achievable by any common resource state consisting of a collection of bipartite entangled states but it is achievable by a common resource state exhibiting multipartite entanglement.
These examples show the difference in the capabilities between these two types of common resource states.

As for the dynamic setting, in addition to considering a limitation on local system sizes for storing a common resource state, the parties are also required preparation of the common resource state within this limitation by performing quantum communication.
Some of the common resource states exhibiting multipartite entanglement analyzed in the static setting can be prepared within the limitation using quantum communication.
Hence, temporal uses of bipartite quantum communication resources are still sufficient for preparing such common resource states.
In contrast, we also show other examples of states exhibiting multipartite entanglement which can be stored but cannot be prepared within a limitation on local system sizes.

The rest of this paper is structured as follows.
In Sec.~\ref{sec:common_resource_state}, we recall the definition of a common resource state for a given target set.
In Sec.~\ref{sec:def}, we introduce the tasks of system-size-limited quantum state preparation in the static setting and the dynamic setting.
We analyze system-size-limited quantum state preparation in the static setting in Sec.~\ref{sec:analysis} and also analyze the dynamic setting in Sec.~\ref{sec:analysis2}.
Our conclusion is given in Sec.~\ref{sec:conclusion}.

\section{\label{sec:common_resource_state}Definition of common resource states}
We begin with recalling the definition of a common resource state for a given set of states under LOCC~\cite{S3,G3}.
In the following, superscripts of an operator or a ket indicate the Hilbert space on which the operator acts or to which the ket belongs.
Note that although we only consider the cases of pure states, generalization to mixed states is straightforward.

To define a target set, consider a quantum system for the states in a target set shared among $N$ parties denoted by $v_1,\ldots,v_N$.
The corresponding Hilbert space is denoted by $\mathcal{H}\coloneqq\bigotimes_{k=1}^{N}\mathcal{H}^{v_k}$, where each party $v_k$'s system is represented by a Hilbert space $\mathcal{H}^{v_k}$.
Let $S$ denote a set of states of $\mathcal{H}$ to be prepared from a common resource state,
which we call the \textit{target set}.
Note that the target set $S$ can be either finite or infinite.

For a given target set $S$ on $\mathcal{H}$,
a common resource state~\cite{S3,G3} is defined as follows.
The total system for a common resource state is denoted by $\overline{\mathcal{H}}\coloneqq\bigotimes_{k=1}^{N}\overline{\mathcal{H}}^{v_k}$, where $\overline{\mathcal{H}}^{v_k}$ for each party $v_k$ denotes the Hilbert space corresponding to $v_k$.
This total system $\overline{\mathcal{H}}$ includes $\mathcal{H}$ for a target set as a subspace, that is, $\overline{\mathcal{H}}^{v_k}\supset\mathcal{H}^{v_k}$ for each party $v_k$.
A state $\Ket\phi\in \overline{\mathcal{H}}$ is called a \textit{common resource state} for the target set $S$ under LOCC if for any state $\Ket{\psi}\in S$ there exists an LOCC protocol which transforms $\Ket\phi$ into $\Ket{\psi}$ deterministically and exactly.
Regarding a formal definition of LOCC, we refer to Ref.~\cite{C2} and the references therein.
Note that the common resource state $\Ket\phi$ for $S$ on $\mathcal{H}$ may only exist in a higher-dimensional Hilbert space $\overline{\mathcal{H}}$ than that for $S$ itself,
that is, $\dim\overline{\mathcal{H}}^{v_k}\geqq\dim\mathcal{H}^{v_k}$ for each $v_k$.

To compare bipartite and multipartite entanglement, we introduce the notion of a common resource state consisting of a collection of bipartite entangled states and that exhibiting multipartite entanglement.
In the following, common resource states are assumed to be fully entangled, that is, entangled with respect to any bipartition of the parties.
Consider a collection of bipartite entangled states distributed among the parties $v_1,\ldots,v_N$.
The distribution of the bipartite entangled states can be represented by a graph $G=(V,E)$, where each vertex in the set $V=\left\{v_1,\ldots,v_N\right\}$ represents a party, and each edge $e=\left\{v_k,v_{k'}\right\}\in E$ a bipartite entangled state $\Ket{\phi_e}^e$ shared between two parties $v_k$ and $v_{k'}$.
A common resource state $\Ket{\phi}$ for a target set $S$ is called a state \textit{consisting of bipartite entanglement} if there exists a graph $G=(V,E)$ such that $\Ket\phi$ is locally unitarily equivalent to a state in the form
$\bigotimes_{e\in E}\Ket{\phi_e}^e$.
Otherwise, $\Ket{\phi}$ is called a state \textit{exhibiting multipartite entanglement}.

For any target set $S$, we can always obtain a common resource state consisting of bipartite entanglement using quantum teleportation~\cite{B4} or a more efficient protocol proposed in Ref.~\cite{Y1}.
This common resource state consists of maximally entangled states distributed among the parties $v_1,\ldots,v_N$ according to a tree $T=(V,E)$, which is a graph including no cycle as a subgraph.
The common resource state can be written as
$\bigotimes_{e\in E}\Ket{\Phi_{M_e}^+}^e$,
where for each edge $e=\{v_k,v_{k^\prime}\}\in E$, $\Ket{\Phi_{M_e}^+}^e\coloneqq\frac{1}{\sqrt{M_e}}\sum_{l=0}^{M_e -1}\Ket{l}^{v_k}\otimes\Ket{l}^{v_{k^\prime}}$
is a maximally entangled state of Schmidt rank $M_e$ shared between $v_k$ and $v_{k'}$.
For any tree $T=(V,E)$, if an edge $e\in E$ is deleted, $T$ is divided into two disjoint trees, whose vertices are represented by disjoint sets $V_e$ and $\overline{V}_e$ satisfying $V=V_e\cup \overline{V}_e$.
For any $\Ket{\psi}\in S$ on $\mathcal{H}=\bigotimes_{k=1}^{N}\mathcal{H}^{v_k}$, we let $R_e\left(\Ket{\psi}\right)$
denote the Schmidt rank of $\Ket{\psi}$ with respect to the bipartition $\bigotimes_{v_k\in V_e}\mathcal{H}^{v_k}$ and $\bigotimes_{v_k\in \overline{V}_e}\mathcal{H}^{v_k}$ of $\mathcal{H}=\bigotimes_{k=1}^{N}\mathcal{H}^{v_k}$.
Given any $\Ket{\psi}\in S$ and a tree $T=(V,E)$, Ref.~\cite{Y1} provides the necessary and sufficient condition for the resource state $\bigotimes_{e\in E}\Ket{\Phi_{M_e}^+}^e$ being transformable into $\Ket{\psi}$ by LOCC\@.
This transformation is achievable if and only if the Schmidt rank of each bipartite maximally entangled state of the resource state is not smaller than the Schmidt rank of $\Ket{\psi}$ with respect to the corresponding bipartition,
that is, for each $e\in E$,
\begin{equation}
    \label{eq:bipartite}
    M_e\geqq R_e\left(\Ket{\psi}\right).
\end{equation}
To obtain a common resource state consisting of bipartite entanglement for $S$, it is sufficient to ensure that the condition of the Schmidt ranks given in Inequality~\eqref{eq:bipartite} is fulfilled for all the states in $S$.

\begin{figure}[t]
    \centering
    \includegraphics[width=3.4in]{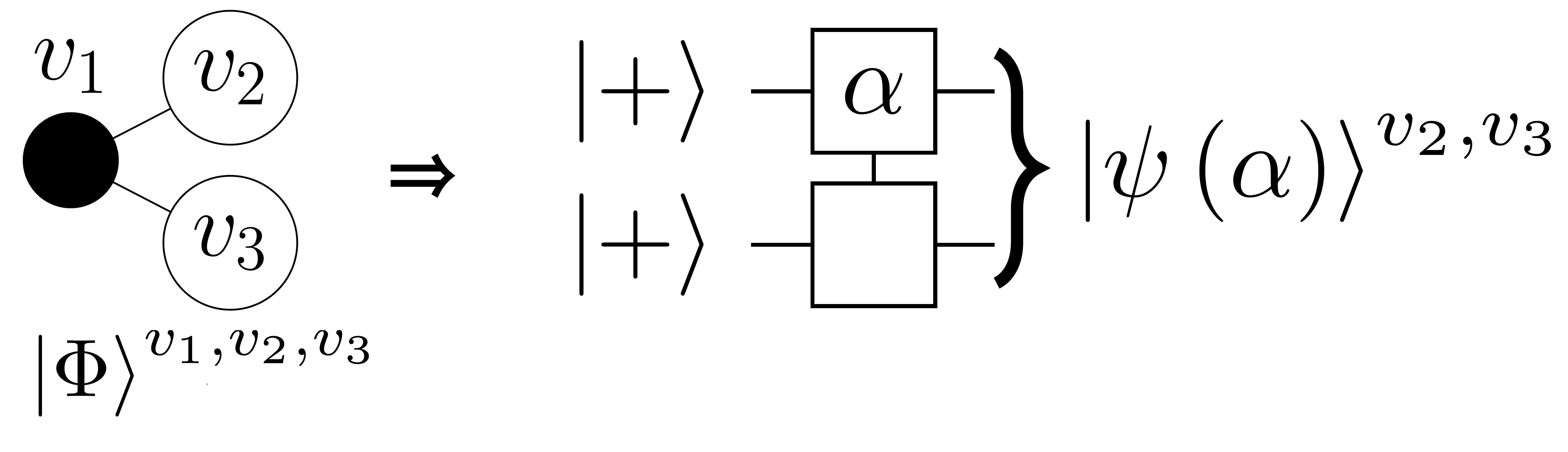}
    \caption{A simple example of a graph representing a graph state and a quantum circuit representing a class of states parameterized by $\alpha$ which can be deterministically prepared using this graph state. Given a graph state $\Ket{\Phi}^{v_1,v_2,v_3}$ as illustrated on the left, by performing the unitary $\exp\left(\textup{i}\alpha X^{v_1}\right)$ parameterized by $\alpha$ and a measurement in the $Z$ basis $\left\{\Ket{0},\Ket{1}\right\}$ on the qubit represented by the black vertex $v_1$, followed by local unitary corrections on the white vertices $v_2$ and $v_3$ conditioned by the measurement outcome, we can deterministically obtain a two-qubit state $\Ket{\psi\left(\alpha\right)}$ defined in Eq.~\eqref{eq:psi_alpha} represented by $v_2$ and $v_3$. The state $\Ket{\psi\left(\alpha\right)}$ can also be represented as the output of the quantum circuit on the right, where a two-qubit gate $\exp\left(\textup{i}\alpha Z^{v_2}\otimes Z^{v_3}\right)$ parameterized by $\alpha$ is applied to $\Ket{+}^{v_2}\otimes\Ket{+}^{v_3}$.}
\label{fig:correspondence}
\end{figure}

As a common resource state exhibiting multipartite entanglement, we can use a class of graph states proposed in Ref.~\cite{S3}.
A graph state~\cite{H4,H5} is a multiqubit entangled state characterized by a graph $G=(V,E)$.
Note that, while graphs in this paper also represent distribution of bipartite entanglement as explained above, a graph state is a different concept, which is a state exhibiting multipartite entanglement obtained for a graph $G=(V,E)$ as follows:
first, for each vertex $v_k\in V$, a qubit labeled $v_k$ is initialized as
\[
    \Ket{+}^{v_k}\coloneqq\frac{1}{\sqrt{2}}\left(\Ket{0}^{v_k}+\Ket{1}^{v_k}\right),
\]
and then, for each edge $e=\left\{v_k,v_{k^\prime}\right\}\in E$, the controlled-$Z$ gate
\begin{equation}
    \label{eq:cz}
    \begin{split}
        &CZ^{v_k,v_{k^\prime}}\\
        &\coloneqq{\left(\Ket{00}\Bra{00}+\Ket{01}\Bra{01}+\Ket{10}\Bra{10}-\Ket{11}\Bra{11}\right)}^{v_k,v_{k^\prime}}
    \end{split}
\end{equation}
is applied to two qubits labeled as $v_k$ and $v_{k^\prime}$.
Reference~\cite{S3} proposes an LOCC protocol for preparing any pure state of an arbitrary number of qubits by performing sequential projective measurements and local unitary corrections on a particular type of graph states. (See also measurement-based quantum computation~\cite{R1,R2,R3}.)
To see how this protocol works, consider the three-vertex graph shown in Fig.~\ref{fig:correspondence} as a simple example.
The graph state $\Ket{\Phi}^{v_1,v_2,v_3}$ represented by this graph is invariant under a local unitary transformation $X^{v_1}\otimes Z^{v_2}\otimes Z^{v_3}$, that is,
\[
    X^{v_1}\otimes Z^{v_2}\otimes Z^{v_3}\Ket{\Phi}^{v_1,v_2,v_3}=\Ket{\Phi}^{v_1,v_2,v_3},
\]
where $X$ and $Z$ are the Pauli operators.
Thus, if the unitary operator $\exp\left(\textup{i}\alpha X^{v_1}\right)$ parameterized by $\alpha$ is performed on qubit $v_1$, the action is equivalent to
\begin{align*}
    &\exp\left(\textup{i}\alpha X^{v_1}\right)\otimes\openone^{v_2}\otimes\openone^{v_3}\Ket{\Phi}^{v_1,v_2,v_3}\\
    &=\openone^{v_1}\otimes\exp\left(\textup{i}\alpha Z^{v_2}\otimes Z^{v_3}\right) \Ket{\Phi}^{v_1,v_2,v_3},
\end{align*}
which can be shown using the Taylor series of the exponential function.
Then, it is straightforward to verify that, performing $\exp\left(\textup{i}\alpha X^{v_1}\right)$ and a measurement in $Z$ basis $\left\{\Ket{0},\Ket{1}\right\}$ on the qubit $v_1$, we obtain a state of two qubits $v_2$ and $v_3$ which can be deterministically transformed by local unitary corrections $\openone^{v_2}\otimes\openone^{v_3}$ or $Z^{v_2}\otimes Z^{v_3}$ conditioned by the measurement outcome $\Ket{0}$ or $\Ket{1}$, respectively, into
\begin{equation}
    \label{eq:psi_alpha}
    \Ket{\psi\left(\alpha\right)}^{v_2,v_3}\coloneqq\exp\left(\textup{i}\alpha Z^{v_2}\otimes Z^{v_3}\right)\left(\Ket{+}^{v_2}\otimes\Ket{+}^{v_3}\right).
\end{equation}
In the same way, it is shown in Ref.~\cite{S3} that any quantum circuit consisting of one-qubit Clifford gates and multiqubit gates $\exp\left(\textup{i}\alpha Z\otimes Z\otimes\cdots\otimes Z\right)$ parameterized by $\alpha$ can be implemented by performing sequential projective measurements and local unitary corrections on a particular graph state corresponding to the quantum circuit.
In addition, it is shown that any pure state of an arbitrary number of qubits is locally unitarily equivalent to a pure state generated by a quantum circuit consisting of these types of gates.
Using this argument, we can obtain a graph state serving as a common resource state exhibiting multipartite entanglement for a given target set.

We remark that the above graph states for common resource states require at least one auxiliary qubit per parameter describing local unitary equivalence classes of $n$-qubit states.
Since the number of parameters of states increases exponentially with respect to $n$,
the size of the required graph states is also exponentially large.
At the same time, if we consider the target set $S$ to be a smaller subset of the $n$ qubits, the number of the parameters describing the states in $S$ can be decreased.
In this case, we may construct another class of graph states which serve as common resource states for the smaller subset $S$, which require a smaller number of auxiliary qubits than the original graph states serving as common resource states for the set of arbitrary $n$-qubit states.
We will use this observation in the subsequent section where the task of system-size-limited quantum state preparation is introduced.

\section{\label{sec:def}Definition of system-size-limited quantum state preparation}

We introduce the tasks of system-size-limited quantum state preparation.
We consider a scenario in which a multipartite system is distributed among spatially separated parties $v_1,\ldots,v_N$, and each party's local system size is limited.
The system-size-limited quantum state preparation for a given target set $S$ is a task for the parties to transform a shared common resource state into an arbitrary state $\Ket{\psi}\in S$ by performing local operations on a limited-size quantum system and classical communication.
To compare multipartite and bipartite resources, we analyze system-size-limited quantum state preparation in two settings---the \textit{static} setting and the \textit{dynamic} setting.
In this section, we first describe the definition of local operations on a limited-size quantum system and then define system-size-limited quantum state preparation in the static setting and the dynamic setting.

To clarify the meaning of local operations on a limited-size quantum system,
we assume that each party $v_k\in\left\{v_1,\ldots,v_N\right\}$ has a quantum system corresponding to a Hilbert space $\overline{\mathcal{H}}^{v_k}$ of dimension
\[
    d^{\left(v_k\right)}\coloneqq\dim\overline{\mathcal{H}}^{v_k}.
\]
The configuration of system sizes for all the parties is denoted by a tuple
\[
    \boldsymbol{d}=\left(d^{\left(v_1\right)},\ldots,d^{\left(v_N\right)}\right).
\]
As explained in Sec.~\ref{sec:common_resource_state}, the target set $S$ is given from a subspace $\mathcal{H}$ of the total system $\overline{\mathcal{H}}$.
Each party $v_k$ can perform any unitary and any measurement on the system $\overline{\mathcal{H}}^{v_k}$
but is \textit{not} allowed to add an auxiliary system to increase the dimension of $\overline{\mathcal{H}}^{v_k}$.
Measurements are represented by quantum instruments,
and while an indirect measurement may require an auxiliary working quantum system,
the protocols in this paper use only projective measurements~\footnote{For the completeness of the definition, we may allow each party $v_k$ to implement an indirect measurement using a projective measurement and one auxiliary working qubit in addition to the system $\overline{\mathcal{H}}^{v_k}$ itself.
This auxiliary working qubit has to be traced out after each measurement.
The use of only one auxiliary working qubit is sufficient for implementing any indirect measurement according to~[E.\ Andersson and D.\ K.\ L.\ Oi, Phys.\ Rev.\ A \textbf{77}, 052104 (2008).].}.
Classical information processing and classical communication without using a quantum system can be freely performed.
For a given configuration of system sizes specified by $\boldsymbol{d}$,
we assume in both the static setting and the dynamic setting that the parties can perform local operations on a limited-size quantum system in the above sense and classical communication.
We refer to this restricted LOCC as \textit{LOCC within the configuration $\boldsymbol{d}$}.

In the dynamic setting, we also allow
any two parties $v_k$ and $v_{k^\prime}$ to perform quantum communication.
When $v_k$ sends a state of a $d$-dimensional system to $v_{k^\prime}$ by quantum communication, $v_k$ has to initially store the state to be sent in a $d$-dimensional subsystem of $\overline{\mathcal{H}}^{v_k}$, and $v_{k^\prime}$ has to initialize a $d$-dimensional subsystem of $\overline{\mathcal{H}}^{v_{k^\prime}}$ as a fixed state $\Ket{0}$ so that $v_{k^\prime}$ receives the state using this subsystem.
After each quantum communication, the $d$-dimensional subsystem of $\overline{\mathcal{H}}^{v_k}$ is initialized as a fixed state $\Ket{0}$ so that $v_k$ can reuse this subsystem.
Each quantum communication from one party to another party is called one \textit{round} of quantum communication.
If a protocol includes multiple rounds of quantum communication, the multiple rounds of quantum communication are performed sequentially.
Quantum communication between the parties is allowed only if it is stated explicitly.

The system-size-limited quantum state preparation in the static setting for a configuration $\boldsymbol{d}$ of system sizes and a target set $S$ is a task for $N$ parties to achieve the following:
\begin{enumerate}
    \item A common resource state $\Ket\phi\in\overline{\mathcal{H}}$ for $S$ is given to the parties;
    \item A particular target state $\Ket{\psi}\in S$ is chosen from the target set $S$, and all the parameters of $\Ket{\psi}$ given to all the parties. Then the parties perform LOCC within the configuration $\boldsymbol{d}$ to transform the common resource state $\Ket\phi$ into this target state $\Ket\psi$.
\end{enumerate}
Our analysis concerns properties of the common resource state $\Ket\phi$ for achieving a system-size-limited quantum state preparation, that is, whether the task is achievable or not when the common resource state $\Ket\phi$ is a state consisting of bipartite entanglement or a state exhibiting multipartite entanglement.

In a similar way, we define system-size-limited quantum state preparation in the dynamic setting as follows.
The system-size-limited quantum state preparation in the dynamic setting for a configuration $\boldsymbol{d}$ of system sizes and a target set $S$ is a task for $N$ parties to achieve the following:
\begin{enumerate}
    \item A common resource state $\Ket\phi\in\overline{\mathcal{H}}$ for $S$ is prepared by the parties using quantum communication in addition to LOCC within the configuration $\boldsymbol{d}$;
    \item A particular target state $\Ket{\psi}\in S$ is chosen from the target set $S$, and all the parameters of $\Ket{\psi}$ given to all the parties. Then the parties perform LOCC within the configuration $\boldsymbol{d}$ to transform the common resource state $\Ket\phi$ into this target state $\Ket\psi$.
\end{enumerate}
In this dynamic setting, $\Ket\phi$ can be a state exhibiting multipartite entanglement as long as $\Ket\phi$ is deterministically prepared by finitely many rounds of quantum communication.
Note that while we differentiate the capabilities of common resource states consisting of bipartite entanglement and those exhibiting multipartite entanglement in the static setting, common resource states in the dynamic setting are expected to have an intermediate capability, since only temporal uses of bipartite quantum communication resources are allowed in the dynamic setting for preparing the common resource states.

In the following, we provide nontrivial examples differentiating between bipartite and multipartite entanglement in the static setting in Sec.~\ref{sec:analysis}.
Also, other examples in the dynamic setting are provided in Sec.~\ref{sec:analysis2} for differentiating the capability of the common resource states in the dynamic setting from that of the common resource states consisting of bipartite entanglement and exhibiting multipartite entanglement in the static setting.

\section{\label{sec:analysis}System-size-limited quantum state preparation in the static setting}
In this section, we analyze system-size-limited quantum state preparation in the static setting.
We show the existence of a system-size-limited quantum state preparation which is achievable by a common resource state exhibiting multipartite entanglement but not by any common resource state consisting of bipartite entanglement.

To show such a nontrivial example, consider eight parties $v_1,\ldots,v_8$.
The configuration $\boldsymbol{d}_0=\left(d_0^{\left(v_1\right)},\ldots, d_0^{\left(v_8\right)}\right)$ of the quantum system sizes are given as follows:
\begin{equation}
    \label{eq:d}
    \begin{split}
        d_0^{\left(v_k\right)}&=\dim\overline{\mathcal{H}}^{v_k} = 4,\; \dim\mathcal{H}^{v_k} = 2, \;\forall v_k\in\{v_1,\ldots,v_7\};\\
        d_0^{\left(v_8\right)}&=\dim\overline{\mathcal{H}}^{v_8} =\dim\mathcal{H}^{v_8} = 2.
    \end{split}
\end{equation}
For each $v_k\in\left\{v_1,\ldots,v_7\right\}$,
we regard the four-dimensional system $\overline{\mathcal{H}}^{v_k}$ as two qubits, where one is for the target set denoted by $\mathcal{H}^{v_k}$ and the other is an auxiliary qubit denoted by $\mathcal{H}_\textup{a}^{v_k}$, that is, $\overline{\mathcal{H}}^{v_k}=\mathcal{H}^{v_k}\otimes\mathcal{H}_\textup{a}^{v_k}$.

We define a target set $S_0$ on $\mathcal{H}=\bigotimes_{k=1}^{N}\mathcal{H}^{v_k}$ as the set of all the possible output states of a quantum circuit illustrated in Fig.~\ref{fig:target_set}.
This circuit consists of seven two-qubit gates $\exp\left(\textup{i}\alpha_i Z\otimes Z\right)$ parameterized by $\alpha_i\in\left\{\alpha_1,\ldots,\alpha_7\right\}$, where $0\leqq\alpha_i < 2\pi$ for each $\alpha_i$.
The tuple of the seven parameters is denoted by
\[
    \boldsymbol\alpha\coloneqq\left(\alpha_1,\ldots,\alpha_7\right).
\]
As input to the circuit, we consider an eight-qubit product state $\Ket{+}^{\otimes 8}\in\mathcal{H}$.
The target set $S_0$ consists of the eight-qubit output states of the circuit parameterized by $\boldsymbol\alpha$, that is,
\begin{equation}
    \label{eq:s}
    S_0\coloneqq\left\{\Ket{\psi\left(\boldsymbol\alpha\right)}\in\mathcal{H}:\boldsymbol\alpha=\left(\alpha_1,\ldots,\alpha_7\right)\right\},
\end{equation}
where each qubit is placed at one of the parties, as illustrated in Fig.~\ref{fig:target_set}.
For example, given the parameters
\[
    \boldsymbol\alpha_0\coloneqq\left(0,0,0,0,0,0,0\right),
\]
each gate in the circuit reduces to the identity and hence $\Ket{\psi\left(\boldsymbol\alpha_0\right)}=\Ket{+}^{\otimes 8}\in S_0$ is a product state.
In contrast, given the parameters
\[
    \boldsymbol\alpha_{\frac{\pi}{4}}\coloneqq\left(\frac{\pi}{4},\frac{\pi}{4},\frac{\pi}{4},\frac{\pi}{4},\frac{\pi}{4},\frac{\pi}{4},\frac{\pi}{4}\right),
\]
$\Ket{\psi\left(\boldsymbol\alpha_\frac{\pi}{4}\right)}\in S_0$ is a fully entangled state, since each gate $\exp\left(\textup{i}\frac{\pi}{4} Z\otimes Z\right)$ entangles $\Ket{+}\otimes\Ket{+}$.

Using the configuration $\boldsymbol{d}_0$ and the target set $S_0$ defined above,
we present the following two propositions on the system-size-limited quantum state preparation for $\boldsymbol{d}_0$ and $S_0$. Proposition~\ref{thm:multipartite} shows the feasibility of the system-size-limited quantum state preparation using a common resource state exhibiting multipartite entanglement while Proposition~\ref{thm:bipartite} is a no-go theorem for any common resource state consisting of bipartite entanglement.

\begin{figure}[t]
    \centering
    \includegraphics[width=3.4in]{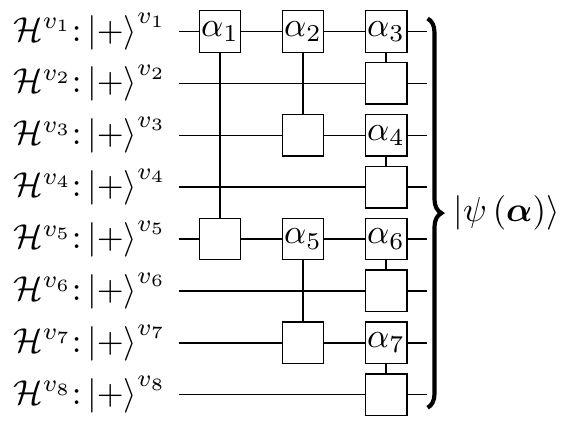}
    \caption{A quantum circuit generating all the states in the target set $S_0\coloneqq\left\{\Ket{\psi\left(\boldsymbol{\alpha}\right)}\right\}$ for the system-size-limited quantum state preparation in Propositions~\ref{thm:multipartite} and~\ref{thm:bipartite}, where $\boldsymbol{\alpha}=\left(\alpha_1,\ldots,\alpha_7\right)$ is a tuple of parameters. The wires of the circuit starting from the input $\Ket{+}^{v_1},\ldots,\Ket{+}^{v_8}$ represent qubits held by the parties $v_1,\ldots,v_8$, respectively. The circuit consists of seven two-qubit gates $\exp\left(\textup{i}\alpha_i Z\otimes Z\right)$ parameterized by $\alpha_i\in\left\{\alpha_1,\ldots,\alpha_7\right\}$.}
\label{fig:target_set}
\end{figure}

\begin{figure}[t]
    \centering
    \includegraphics[width=3.4in]{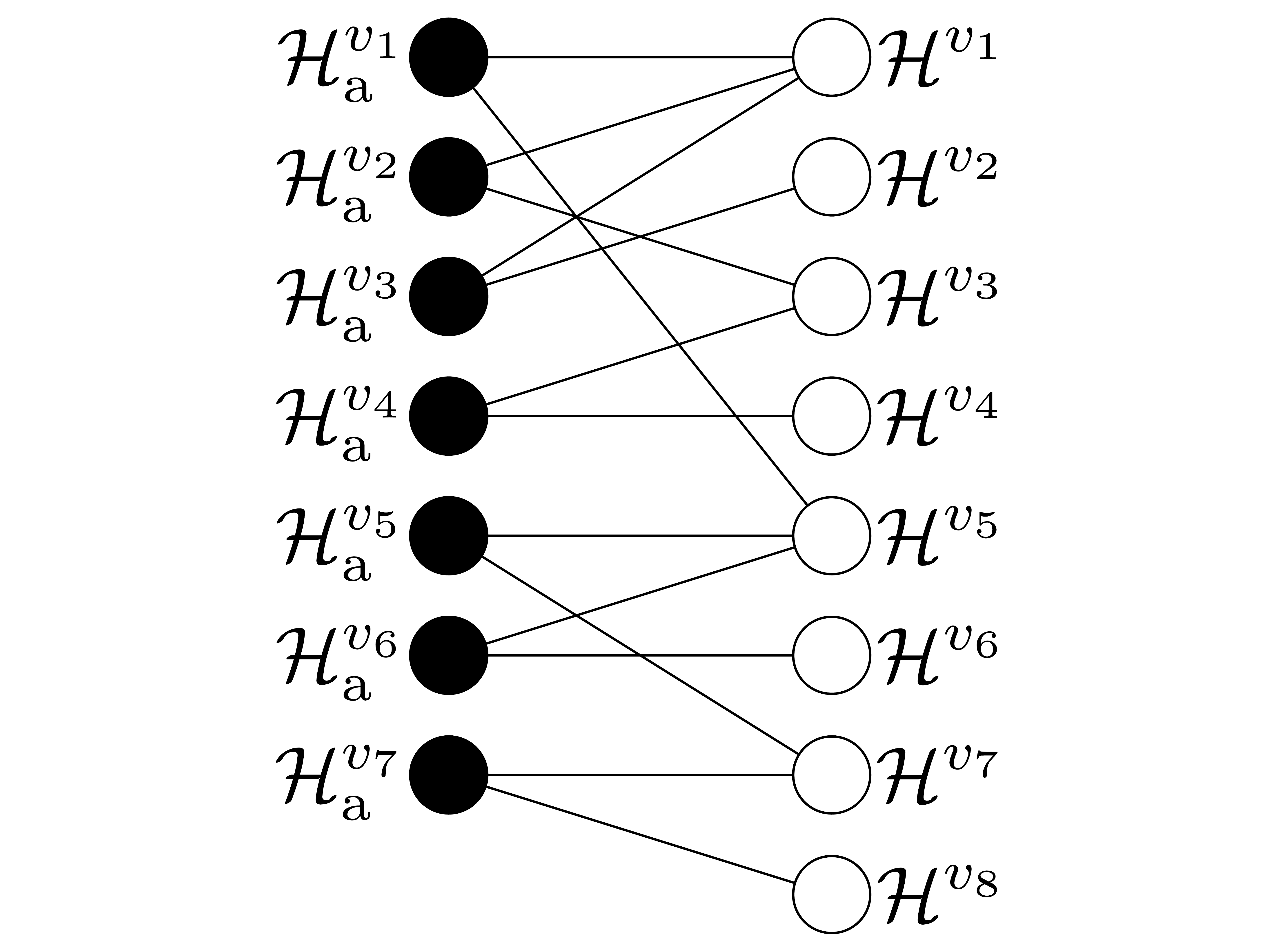}
    \caption{A graph representing a $15$-qubit graph state $\Ket{\Phi_\textup{res}}$ used as a common resource state exhibiting multipartite entanglement in Proposition~\ref{thm:multipartite}.
    Each of the parties $v_k\in\left\{v_1,\ldots,v_7\right\}$ holds two qubits $\mathcal{H}^{v_k}\otimes\mathcal{H}_\textup{a}^{v_k}$, while party $v_8$ holds one qubit $\mathcal{H}^{v_8}$.
    Eight of the $15$ qubits $\mathcal{H}^{v_1},\ldots,\mathcal{H}^{v_8}$ represented by white vertices are qubits which can be prepared in any state $\Ket{\psi\left(\boldsymbol{\alpha}\right)}$ in the target set $S_0$.
    The other seven $\mathcal{H}_\textup{a}^{v_1},\ldots,\mathcal{H}_\textup{a}^{v_7}$ represented by black vertices are auxiliary qubits to be measured.
    To obtain $\Ket{\psi\left(\boldsymbol{\alpha}\right)}\in S_0$ parameterized by $\boldsymbol\alpha=\left(\alpha_1,\ldots\alpha_7\right)$,
    each party $v_i\in\left\{v_1,\ldots,v_7\right\}$ performs the following protocol in order.
    First, a unitary $\exp\left(\textup{i}\alpha_i X\right)$ parameterized by $\alpha_i$ is performed on $\mathcal{H}_\textup{a}^{v_i}$.
     Then, the qubit $\mathcal{H}_\textup{a}^{v_i}$ is measured in the $Z$ basis $\left\{\Ket{0},\Ket{1}\right\}$, and depending on the outcome, a unitary correction is applied to the remaining qubits.
     Using this protocol, the parties can deterministically transform $\Ket{\Phi_\textup{res}}$ into $\Ket{\psi\left(\boldsymbol{\alpha}\right)}\in S_0$ for any $\boldsymbol\alpha$.}
\label{fig:tree}
\end{figure}

\begin{proposition}
\label{thm:multipartite}
    \textit{Multipartite entanglement in a system-size-limited quantum state preparation in the static setting.}
    The system-size-limited quantum state preparation in the static setting for the configuration $\boldsymbol{d}_0$ defined in Eq.~\eqref{eq:d} and the target set $S_0$ defined in Eq.~\eqref{eq:s}
    is achievable using a common resource state exhibiting multipartite entanglement.
\end{proposition}

\begin{proposition}
\label{thm:bipartite}
    \textit{Bipartite entanglement in a system-size-limited quantum state preparation in the static setting.}
    The system-size-limited quantum state preparation in the static setting for the configuration $\boldsymbol{d}_0$ defined in Eq.~\eqref{eq:d} and the target set $S_0$ defined in Eq.~\eqref{eq:s}
    is \textit{not} achievable using any common resource state consisting only of bipartite entanglement.
\end{proposition}

In the following, we prove Propositions~\ref{thm:multipartite} and~\ref{thm:bipartite}.
Note that while shallower quantum circuits having a similar structure to the circuit in Fig.~\ref{fig:target_set} are not sufficient for differentiating between multipartite and bipartite entanglement, this example might not be the simplest, and further sets of states with the same properties will also be given after the proofs.

\begin{proof}[Proof of Proposition~\ref{thm:multipartite}]
    We provide a common resource state exhibiting multipartite entanglement for the target set $S_0$,
    which is the $15$-qubit graph state $\Ket{\Phi_\textup{res}}$ illustrated in Fig.~\ref{fig:tree} held by the parties $v_1,\ldots,v_8$.
    In the same way as explained in Sec.~\ref{sec:common_resource_state}, given the graph state $\Ket{\Phi_\textup{res}}$ in Fig.~\ref{fig:tree},
    for each $i\in\left\{1,\ldots,7\right\}$,
    performing $\exp\left(\textup{i}\alpha_i X^{v_i}\right)$ parameterized by $\alpha_i$ and a measurement in the $Z$ basis $\left\{\Ket{0},\Ket{1}\right\}$ on the qubit represented by $\mathcal{H}_\textup{a}^{v_i}$, followed by local unitary corrections on other qubits conditioned by the measurement outcome,
    the parties can obtain $\Ket{\psi\left(\boldsymbol\alpha\right)}\in S_0$ deterministically for any parameters $\boldsymbol\alpha=\left(\alpha_1,\ldots,\alpha_7\right)$.
\end{proof}

\begin{proof}[Proof of Proposition~\ref{thm:bipartite}]
    We derive a necessary condition for preparing the state $\Ket{\psi\left(\boldsymbol\alpha_{\frac{\pi}{4}}\right)}\in S_0$ from a resource state consisting of bipartite entanglement by LOCC within the configuration $\boldsymbol d_0$.
    Observe that the state $\Ket{\psi\left(\boldsymbol\alpha_{\frac{\pi}{4}}\right)}$ is fully entangled, that is, entangled with respect to any bipartition of the eight qubits.
    To prepare a fully entangled state, the resource state at party $v_8$ has to be entangled with some other parties.
    As $\dim \overline{\mathcal{H}}^{v_8}=2$, the party $v_8$ can store only one qubit of a bipartite resource state entangled with another party, which we label as $u_7\in\{v_1,\ldots,v_7\}$.
    The quantum system $\overline{\mathcal{H}}^{u_7}$ at $u_7$ is decomposed into $\overline{\mathcal{H}}^{u_7}=\mathcal{H}^{u_7}_{\{u_7,v_8\}}\otimes\mathcal{H}^{u_7}_\textup{r}$ where $\mathcal{H}^{u_7}_{\{u_7,v_8\}}$ is a system of more than one dimension for the bipartite entangled resource state shared with $v_8$, and $\mathcal{H}^{u_7}_\textup{r}$ the remaining quantum system.
    It is necessary that
    \begin{equation}
        \label{eq:dim}
        \begin{split}
            &\dim\mathcal{H}^{u_7}_{\{u_7,v_8\}}=2,\\
            &\dim\mathcal{H}^{u_7}_\textup{r}=2,
        \end{split}
    \end{equation}
    which can be shown by contradiction as follows.
    Assume that $\dim\mathcal{H}^{u_7}_{\{u_7,v_8\}}>2$.
    Then we have $\dim\mathcal{H}^{u_7}_\textup{r}<2$, and the resource state shared between the parties $u_7$ and $v_8$ cannot be entangled with any of the other parties.
    This contradicts the assumption that a fully entangled state can be prepared, and Eq.~\eqref{eq:dim} is shown.
    As $\dim\mathcal{H}^{u_7}_\textup{r}=2$, the party $u_7$ can store another single qubit of a bipartite resource state entangled with a party other than $v_8$, which we label as $u_6\in\{v_1,\ldots,v_7\}\setminus\{u_7\}$.
    By iterating the above argument, any resource state consisting of bipartite entanglement for preparing a fully entangled state by LOCC within the configuration $\boldsymbol{d}_0$ is required to be seven two-qubit entangled states shared between $u_1$--$u_2$, $\ldots$, $u_6$--$u_7$, and $u_7$--$v_8$, respectively, where
    \begin{equation}
        \label{eq:perm}
        \begin{split}
        (u_1,\ldots,u_7)~\text{is a permutation of}~(v_1,\ldots,v_7).
        \end{split}
    \end{equation}
    Note that although $u_1$ uses only one qubit in this case, the remaining system of $u_1$, which is two dimensional, cannot be used for sharing an entangled state with the other parties since there is no dimension left in the quantum systems of the other parties.
    Therefore, the distribution of the two-qubit entangled states is represented by a line-topology graph, as illustrated in Fig.~\ref{fig:permutation}.
    Note that this line-topology graph is a tree.
    Since the target set $S_0$ includes a fully entangled state $\Ket{\psi\left(\boldsymbol\alpha_\frac{\pi}{4}\right)}$,
    it is necessary that any common resource state consisting of bipartite entanglement for $S_0$ within the configuration $\boldsymbol d_0$ is a state consisting of seven two-qubit entangled states represented by the line-topology tree as shown in Fig.~\ref{fig:permutation}.

\begin{figure}[t]
    \centering
    \includegraphics[width=3.4in]{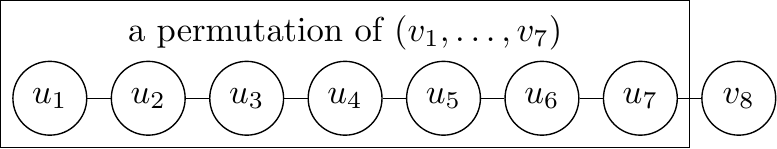}
    \caption{A line-topology tree representing a resource state consisting of bipartite entanglement to prepare a fully entangled state within the configuration $\boldsymbol{d}_0$. Since the target set $S_0$ includes a fully entangled state $\Ket{\psi\left(\boldsymbol\alpha_\frac{\pi}{4}\right)}$, the common resource states consisting of bipartite entanglement for $S_0$ have to be represented by the line-topology tree in the figure, which leads to a contradiction with the condition given in Inequality~\eqref{eq:bipartite} as shown in the main text.}
\label{fig:permutation}
\end{figure}

    We prove that the state $\Ket{\psi(\boldsymbol\alpha_\frac{\pi}{4})}$ cannot be prepared from any such resource state.
    Since any two-qubit entangled state can be obtained by LOCC from a Bell state $\frac{1}{\sqrt{2}}\left(\Ket{00}+\Ket{11}\right)$,
    it suffices to consider resource states consisting of seven Bell states represented by the line-topology tree.
    Thus, the condition given in Inequality~\eqref{eq:bipartite} implies that the state $\Ket{\psi\left(\boldsymbol\alpha_\frac{\pi}{4}\right)}$ can be prepared from resource states consisting of seven Bell states represented by a line-topology tree if and only if
    \[
        R_e\left(\Ket{\psi\left(\boldsymbol\alpha_\frac{\pi}{4}\right)}\right)\leqq 2
    \]
    for any edge $e$ of the line-topology tree.
    In other words, the Schmidt rank of $\Ket{\psi\left(\boldsymbol\alpha_\frac{\pi}{4}\right)}$ with respect to each edge of the line-topology tree needs to be smaller than or equal to $2$.
    However, the explicit calculation of $R_e\left(\Ket{\psi(\boldsymbol\alpha_\frac{\pi}{4})}\right)$ for all the edges $e$ of all the $7!=5040$ different trees obtained from the permutations of $v_1,\ldots,v_7$ in Eq.~\eqref{eq:perm} shows that, for any of the permutations, there exists an edge $e$ such that
    \begin{equation}
        \label{eq:r_e}
        R_e\left(\Ket{\psi\left(\boldsymbol\alpha_\frac{\pi}{4}\right)}\right)>2.
    \end{equation}
    For details, see Appendix~\ref{sec:calculation}.
    The calculation of $R_e\left(\Ket{\psi\left(\boldsymbol\alpha_\frac{\pi}{4}\right)}\right)$ implies that the state $\Ket{\psi(\boldsymbol\alpha_\frac{\pi}{4})}$ cannot be prepared from any resource state consisting of the seven Bell states.
    Therefore, we conclude that there exists no common resource state consisting of bipartite entanglement for the target set $S_0$ within the configuration $\boldsymbol d_0$.
\end{proof}

We remark that, regarding the required system sizes for storing a common resource state consisting of bipartite entanglement,
the above argument for Proposition~\ref{thm:bipartite} based on the Schmidt ranks of a common resource state consisting of bipartite entanglement can also apply to more general target sets.
In particular, we analyze in Appendix~\ref{sec:general} a target set $S$ of $2m$-qudit states, where the size of each qudit is $d$, and each state in $S$ has maximal Schmidt ranks with respect to any bipartition between $m$ qudits and the other $m$ qudits.
Note that random weighted graph states or random pure states fulfill this condition, for which the reduced states have almost maximum entropy for any bipartition~\cite{C4}.
We show in Appendix~\ref{sec:general} that for any resource state consisting of bipartite entanglement to obtain such a state by LOCC, or even by stochastic LOCC, there has to be at least one party for which the local quantum system size for storing this resource state needs to be almost quadratically larger than $d$, that is, greater than or equal to $d^{2-\frac{1}{m}}$.
We also note that for some special configurations of local system sizes, these differences between bipartite and multipartite entanglement do not arise, especially if $\dim\overline{\mathcal{H}}^{v_1}\geqq\prod_{k=2}^{N}\dim\overline{\mathcal{H}}^{v_k}$~\cite{R5}.

\section{\label{sec:analysis2}System-size-limited quantum state preparation in the dynamic setting}
In this section, we analyze the difference in system-size-limited quantum state preparation between the static setting and the dynamic setting.
Before analyzing multipartite cases, we first discuss a simpler bipartite case to clarify the difference  between the static setting and the dynamic setting.
Consider two parties $v_1$ and $v_2$, where each party has two qubits; that is, the configuration $\left(d^{\left(v_1\right)},d^{\left(v_2\right)}\right)$ is given by
\begin{align*}
    d^{\left(v_1\right)}&=\dim\overline{\mathcal{H}}^{v_1} = 4,\\
    d^{\left(v_2\right)}&=\dim\overline{\mathcal{H}}^{v_1} = 4.
\end{align*}
In this case, these two parties can store an entangled resource state of Schmidt rank $4$ in the static setting.
However, in the dynamic setting, the parties can prepare an entangled resource state of Schmidt rank at most $2$, which is shown as follows.
Consider any shared state $\Ket{\phi}^{v_1,v_2}$ after the last round of quantum communication for preparing $\Ket{\phi}^{v_1,v_2}$, where we assume that the direction of the quantum communication in the last round is from $v_1$ to $v_2$ without loss of generality.
Since the quantum communication sends out at least one qubit from $v_1$, the rank of $v_1$'s reduced state for $\Ket{\phi}^{v_1,v_2}$ is at most $2$; that is, the Schmidt rank of $\Ket{\phi}^{v_1,v_2}$ is at most two.
Since the Schmidt rank is monotonically nonincreasing by LOCC~\cite{L}, $v_1$ and $v_2$ after the last round of quantum communication cannot prepare an entangled resource state of Schmidt rank more than $2$, which yields the conclusion.
Although this two-party example is trivial, we also show nontrivial cases of more than two parties in the following.

We present the following two propositions.
Proposition~\ref{prp:bipartite_dynamic} shows that the common resource states available in the dynamic setting can still have more capability than any common resource state consisting of bipartite entanglement in the static setting, similarly to the common resource states exhibiting multipartite entanglement in the static setting.
In contrast, Proposition~\ref{prp:multipartite_dynamic} shows the existence of common resource states which cannot be prepared in the dynamic setting by the parties within a limitation of local system sizes, while the common resource states can still be stored within the limitation in the static setting.
This implies that the common resource states in the dynamic setting have in this case less capability than a common resource state exhibiting multipartite entanglement in the static setting.

\begin{figure*}[t]
    \centering
    \includegraphics[width=7.0in]{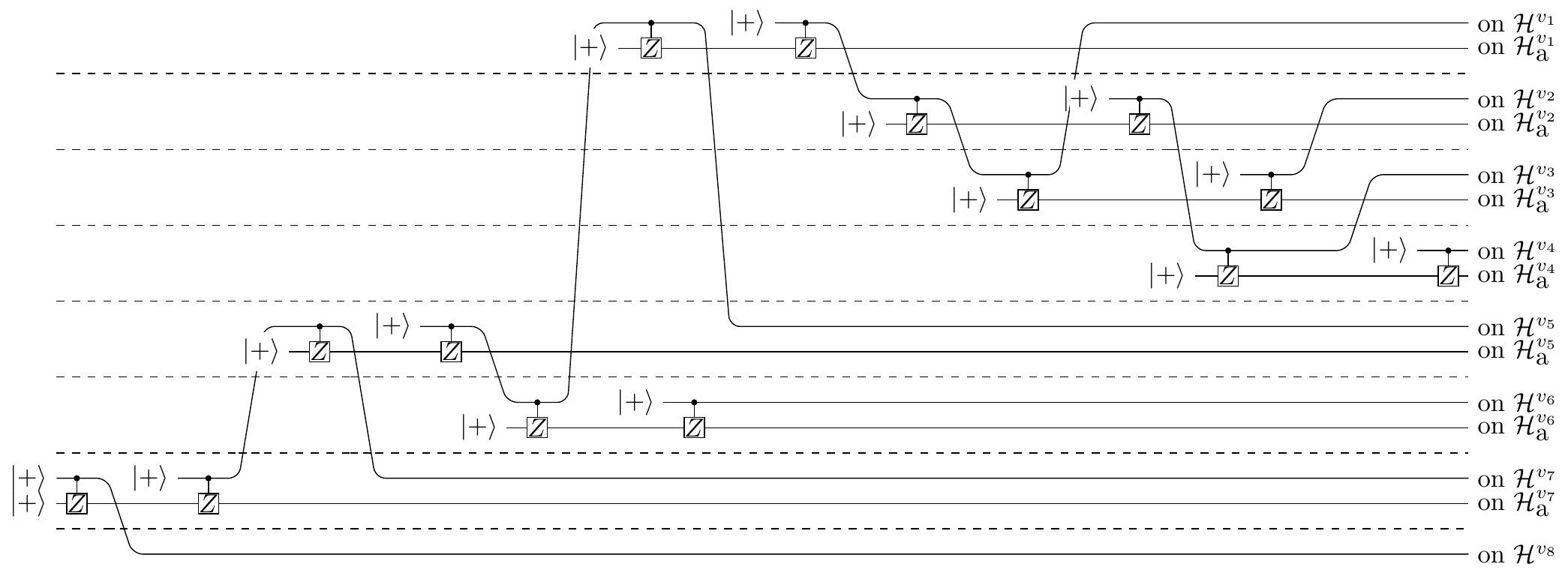}
    \caption{A quantum circuit representing a protocol for preparing the common resource state $\Ket{\Phi_\textup{res}}$ for the target set $S_0$ by quantum communication in addition to LOCC within the configuration $\boldsymbol{d}_0$. Each of the parties $v_k\in\left\{v_1,\ldots v_7\right\}$ can perform local operations on at most two qubits $\mathcal{H}^{v_k}\otimes\mathcal{H}_\textup{a}^{v_k}$, while the party $v_8$ can perform local operations on one qubit $\mathcal{H}^{v_8}$. The dashed lines represent the separation of the parties. Each wire of the circuit corresponds to a qubit corresponding to the Hilbert space on the right, and the circuit consists of the controlled-$Z$ gates $CZ$ defined in Eq.~\eqref{eq:cz} and quantum communication represented by crossings of the wires.}
\label{fig:multipartite}
\end{figure*}

\begin{proposition}
\label{prp:bipartite_dynamic}
    \textit{A common resource state in the dynamic setting having more capability than any common resource state consisting of bipartite entanglement.}
    The state $\Ket{\Phi_\textup{res}}$ in the proof of Proposition~\ref{thm:multipartite} and in Fig.~\ref{fig:tree} can be used as a common resource state for achieving the system-size-limited quantum state preparation in the dynamic setting for the configuration $\boldsymbol{d}_0$ defined in Eq.~\eqref{eq:d} and the target set $S_0$ defined in Eq.~\eqref{eq:s}, while the system-size-limited quantum state preparation in the static setting for $\boldsymbol{d}_0$ and $S_0$ cannot be achieved by any common resource state consisting of bipartite entanglement due to Proposition~\ref{thm:bipartite}.
\end{proposition}

\begin{proposition}
\label{prp:multipartite_dynamic}
    \textit{Common resource states exhibiting multipartite entanglement which cannot be prepared in the dynamic setting.}
    Consider four parties $v_1$, $v_2$, $v_3$, and $v_4$.
    Given a configuration $\boldsymbol{d}_1=\left(d_1^{\left(v_1\right)},d_1^{\left(v_2\right)},d_1^{\left(v_3\right)},d_1^{\left(v_4\right)}\right)$, where
    \begin{align*}
        d_1^{\left(v_1\right)}&=\dim\overline{\mathcal{H}}^{v_1} = 4,\\
        d_1^{\left(v_k\right)}&=\dim\overline{\mathcal{H}}^{v_k} = 2, \; \forall v_k\in\left\{v_2,v_3,v_4\right\},
    \end{align*}
    any fully entangled common resource state $\Ket{\phi}\in\overline{\mathcal{H}}$ whose Schmidt rank with respect to the bipartition between $v_1$ and $v_2 v_3 v_4$ is more than $2$ cannot be prepared in the dynamic setting, although there exists such a common resource state which can be stored in the static setting.
\end{proposition}

First, we prove Proposition~\ref{prp:bipartite_dynamic} as follows.
\begin{proof}[Proof of Proposition~\ref{prp:bipartite_dynamic}]
    We show that the common resource state $\Ket{\Phi_\textup{res}}$ in the proof of Proposition~\ref{thm:multipartite} and in Fig.~\ref{fig:tree} can be prepared by the parties using quantum communication in addition to LOCC within the configuration $\boldsymbol{d}_0$.
    The protocol for preparing $\Ket{\Phi_\textup{res}}$ is represented by a quantum circuit illustrated in Fig.~\ref{fig:multipartite}.
    In this circuit, the parties repeatedly perform $CZ$ gates defined in Eq.~\eqref{eq:cz} to entangle qubits initialized as $\Ket{+}$, distribute one qubit of the entangled state by quantum communication, and perform a $CZ$ gate again to entangle the remaining part of the entangled state with another qubit initialized as $\Ket{+}$.
    After this protocol, the state $\Ket{\Phi_\textup{res}}$ is shared among the parties $v_1,\ldots,v_8$.
\end{proof}

Next, we prove Proposition~\ref{prp:multipartite_dynamic} in a similar way to the example given at the beginning of this section.
\begin{proof}[Proof of Proposition~\ref{prp:multipartite_dynamic}]
Consider any fully entangled state $\Ket{\phi}^{v_1,v_2,v_3,v_4}$ shared among $v_1$, $v_2$, $v_3$, and $v_4$ after the last round of quantum communication for preparing $\Ket{\phi}^{v_1,v_2,v_3,v_4}$.
The direction of the quantum communication in the last round is either of the following three possibilities:
\begin{enumerate}
    \item From $v_1$ to $v_k$ where $k\in\left\{2,3,4\right\}$;
    \item From $v_k$ to $v_{k^\prime}$ where $k,k^\prime\in\left\{2,3,4\right\}$ and $k\neq k^\prime$;
    \item From $v_k$ to $v_1$ where $k\in\left\{2,3,4\right\}$.
\end{enumerate}
Since $\Ket{\phi}^{v_1,v_2,v_3,v_4}$ is fully entangled, we exclude the latter two possibilities 2 and 3, which lead to a product state between $v_k$ and the others.
Regarding possibility 1, after sending at least one qubit from $v_1$ to $v_k$, the rank of $v_1$'s reduced state for $\Ket{\phi}^{v_1,v_2,v_3,v_4}$ is at most $2$; that is, the Schmidt rank of $\Ket{\phi}^{v_1,v_2,v_3,v_4}$ with respect to the bipartition between $v_1$ and $v_2 v_3 v_4$ is at most $2$.
Since the Schmidt rank is monotonically nonincreasing by LOCC~\cite{L}, the parties after the last round of quantum communication cannot prepare any common resource state whose Schmidt rank with respect to the bipartition between $v_1$ and $v_2 v_3 v_4$ is more than $2$, which yields the conclusion.
\end{proof}

Note that under the limitation in Proposition~\ref{prp:multipartite_dynamic}, the parties can prepare any state whose Schmidt rank with respect to the bipartition between $v_1$ and $v_2 v_3 v_4$ is not more than $2$.
This is because $v_1$'s reduced state can be represented by one qubit in this case, and hence, the parties can perform quantum communication to bring arbitrary two qubits to $v_1$ to perform any two-qubit gates.
We also remark that, while we assume in our analysis that quantum communication is performed sequentially, one can also consider simultaneous quantum communication between two parties, which is considered as a swap operation between the two.
However, this simultaneous quantum communication yields a trivial result, since the parties under the limitation in Proposition~\ref{prp:multipartite_dynamic} can prepare any state $\Ket{\Phi}\in\overline{\mathcal{H}}$ using swap operations for letting $v_1$ perform arbitrary two-qubit gates.

\section{\label{sec:conclusion}Conclusion}
We introduced and analyzed the task of system-size-limited quantum state preparation for comparing multipartite and bipartite entanglement from the viewpoint of state convertibility by local operations and classical communication (LOCC).
In contrast to previous studies on the LOCC convertibility between multipartite pure states of the \textit{same}-dimensional systems~\cite{V2,S1,S2,H2,G2,S4,T,T2}, we analyzed LOCC transformation from a common resource state~\cite{S3,G3} of a \textit{higher}-dimensional Hilbert space into a set of states of a \textit{lower}-dimensional Hilbert space.

Introducing a limitation on the size of the local system of each party, we analyzed the capabilities of common resource states exhibiting multipartite entanglement and those consisting of bipartite entanglement.
By showing a nontrivial example, we differentiate the capabilities of these common resource states in terms of achievability of the system-size-limited quantum state preparations for the same target set in the static setting where a common resource state has to be stored within a given limitation of local system sizes.
In addition to this static setting, we considered the dynamic setting where the parties may use a common resource state exhibiting multipartite entanglement, but this common resource state has to be prepared by temporal uses of bipartite quantum communication resources within the limitation of local system sizes.
We also provided nontrivial examples implying that common resource states in the dynamic setting have an intermediate capability between the common resource states exhibiting multipartite entanglement and those consisting of bipartite entanglement.

Our results provide examples implying that multipartite entanglement outperforms bipartite entanglement when limitations on the local system sizes matter in both the static setting and the dynamic setting.
Further research will be needed to establish more general connections between the system sizes for common resource states and properties differentiating multipartite and bipartite entanglement.

\acknowledgments{%
    This work was supported by JSPS Overseas Challenge Program for Young Researchers, Grant-in-Aid for JSPS Research Fellow, JSPS KAKENHI (Grant Numbers 26330006, 15H01677, 16H01050, 17H01694, 18H04286, and 18J10192), and the Austrian Science Fund (FWF) (P28000-N27, P30937-N27, and Y535-N16).
}

\appendix

\section{\label{sec:calculation}How to calculate the ranks in Inequality~\eqref{eq:r_e}}
We show that the Schmidt rank $R_e\left(\Ket{\psi(\boldsymbol\alpha_\frac{\pi}{4})}\right)$ in Inequality~\eqref{eq:r_e} can be exactly calculated with the help of a computer program.
Although computers cannot calculate irrational numbers exactly, we can reduce the Schmidt rank $R_e\left(\Ket{\psi(\boldsymbol\alpha_\frac{\pi}{4})}\right)$ of a vector $\Ket{\psi(\boldsymbol\alpha_\frac{\pi}{4})}$ with irrational elements to that of a vector only with integer elements.
To remove irrational coefficients for normalization of the state $\Ket{+}$ and the gates $\exp\left(\textup{i}\frac{\pi}{4}Z\otimes Z\right)$,
we substitute $\Ket{+}$ and $\exp\left(\textup{i}\frac{\pi}{4}Z\otimes Z\right)$ in the circuit in Fig.~\ref{fig:target_set} with $\sqrt{2}\Ket{+}$ and $\sqrt{2}\exp\left(\textup{i}\frac{\pi}{4}Z\otimes Z\right)$, respectively.
The resulting vector
\[
    \Ket{\tilde\psi\left(\boldsymbol\alpha_\frac{\pi}{4}\right)}\coloneqq 2^{\frac{15}{2}}\Ket{\psi\left(\boldsymbol\alpha_\frac{\pi}{4}\right)}
\]
has the same Schmidt ranks as $\Ket{\psi\left(\boldsymbol\alpha_\frac{\pi}{4}\right)}$ for any bipartition, and all the elements of $\Ket{\tilde\psi\left(\boldsymbol\alpha_\frac{\pi}{4}\right)}$ are complex numbers whose real and imaginary parts are both integers by construction.
Therefore, we can exactly calculate Schmidt ranks of $\Ket{\psi\left(\boldsymbol\alpha_\frac{\pi}{4}\right)}$ by calculating those of $\Ket{\tilde\psi\left(\boldsymbol\alpha_\frac{\pi}{4}\right)}$ by computer.

\section{\label{sec:general}Requirement for common resource states consisting of bipartite entanglement for preparing states having maximal Schmidt ranks}
We show the following proposition.
\begin{proposition}
    \textit{Requirement for resource states consisting of bipartite entanglement for preparing a multipartite entangled state having maximal Schmidt ranks.}
    Consider a $2m$-qudit state $\Ket{\psi}\in\mathcal{H}\coloneqq{\left(\mathbb{C}^d\right)}^{\otimes 2m}$ of local system size $d$ which has maximal Schmidt rank with respect to bipartite cuts between any $m$ qudits and the other $m$ qudits; that is, for any such bipartite cut, the Schmidt rank is $d^m$.
    If $2m$ parties $v_1,\ldots,v_{2m}$ prepare $\Ket{\psi}$ by LOCC from any resource state consisting only of bipartite entanglement,
    then there has to exist at least one party $v\in V\coloneqq\left\{v_1,\ldots,v_{2m}\right\}$ for which the local system size $\dim\overline{\mathcal{H}}^{v}$ for storing this resource state is almost quadratically larger, that is,
    \begin{equation}
        \label{eq:full_schmidt_rank}
        \max_{v_k\in V} \left\{\dim\overline{\mathcal{H}}^{v_k}\right\} \geqq d^{2- \frac{1}{m}}.
    \end{equation}
\end{proposition}

\begin{proof}
    Since any bipartite state can be obtained from a maximally entangled state, it suffices to evaluate $\dim\overline{\mathcal{H}}^{v_k}$ for storing a resource state consisting of bipartite maximally entangled states distributed according to the complete graph $K=\left(V,E\right)$, that is, the fully-connected graph for the $2m$ parties.
    We let $M_e\in\left\{1,2,\ldots\right\}$ denote the Schmidt rank of the maximally entangled state for each edge $e\in E$.

    We first derive a lower bound of the total system size for storing $\bigotimes_{e\in E}\Ket{\Phi_{M_e}^+}^e$, that is, $\prod_{e\in E}{\left(M_e\right)}^2$.
    Consider an edge cut $C$~\cite{B7} of $K$ between any $m$ vertices and the other $m$ vertices.
    Since the Schmidt rank is monotonically nonincreasing under LOCC~\cite{L},
    it is necessary that, for any $C$,
    \begin{equation}
        \label{eq:schmidt_rank_condition}
        \prod_{e\in C}M_e\geqq d^m.
    \end{equation}
    Considering Inequality~\eqref{eq:schmidt_rank_condition} for all the ${{2m\choose m}}/2$ possible choices of $C$ between any $m$ vertices and the other $m$ vertices, and taking the products of the right- and left-hand sides of these inequalities, we obtain
    \[
        \prod_{C} \prod_{e\in C} M_e\geqq d^{m\frac{{2m\choose m}}{2}}.
    \]
    Since $M_e$ for each $e \in  E$ appears ${2m-2\choose m-1}$ times in the product on the left-hand side, the last inequality can be written as
    \[
        \prod_{C} \prod_{e\in C} M_e =\prod_{e\in E}{\left(M_e\right)}^{{2m-2}\choose{m-1}}\geqq d^{m\frac{{2m\choose m}}{2}}.
    \]
    Therefore, a lower bound of the total system size is
    \[
        \prod_{e\in E}{\left(M_e\right)}^2\geqq d^{2\left(2m-1\right)}.
    \]

    Since the total system size for storing $\bigotimes_{e\in E}\Ket{\Phi_{M_e}^+}^e$ is written as
    \[
        \dim\overline{\mathcal{H}}=\prod_{v_k\in V}\dim\overline{\mathcal{H}}^{v_k},
    \]
    we have
    \[
        \prod_{v_k\in V}\dim\overline{\mathcal{H}}^{v_k}\geqq\prod_{e\in E}{\left(M_e\right)}^2\geqq d^{2\left(2m-1\right)}.
    \]
    Therefore, we obtain
    \[
        \max_{v_k\in V} \left\{\dim\overline{\mathcal{H}}^{v_k}\right\}\geqq{\left(\prod_{v_k\in V}\dim\overline{\mathcal{H}}^{v_k}\right)}^{\frac{1}{2m}}\geqq d^{2-\frac{1}{m}},
    \]
    which yields the conclusion.
\end{proof}

Note that the lower bound of local system sizes in Inequality~\eqref{eq:full_schmidt_rank} is almost sufficient for fulfilling the condition~\eqref{eq:schmidt_rank_condition} of the Schmidt ranks by storing a symmetric distribution of maximally entangled states shared between all pairs of the parties.
In this case, since each party shares maximally entangled states with the other $2m-1$ parties, the maximally entangled state corresponding to each $e\in E$ satisfies $M_e=\left\lceil d^\frac{1}{m}\right\rceil$, and the local system size for each $v\in V$ is ${\left\lceil d^\frac{1}{m}\right\rceil}^{2m-1}$, where $\lceil{}\cdots{}\rceil$ is the ceiling function.

\end{document}